\RequirePackage{fix-cm}
\documentclass[smallextended]{svjour3}       
\smartqed  
\usepackage{graphicx}
\usepackage{amsmath}
\usepackage{amsfonts}
\usepackage{amssymb}
\usepackage{bbm}
\newcommand{\ud}{\mathrm{d}}

\begin{document}

\title{The value of power-related options under spectrally negative L\'evy processes
}


\author{Jean-Philippe Aguilar
}


 \institute{J.~Ph. Aguilar  \at
               Cov\'ea Finance - Quantitative Research Team - 8-12 Rue Boissy d'Anglas - FR75008 Paris \\
               \email{jean-philippe.aguilar@covea-finance.fr}          
}

\date{This version: Nov. $19^{\mathrm{th}}$, 2019, revised Sept. $27^{\mathrm{th}}$, 2020.} 

\maketitle

\begin{abstract}
We provide analytical tools for pricing power options with exotic features (capped or log payoffs, gap options etc.) in the framework of exponential L\'evy models driven by one-sided stable or tempered stable processes. Pricing formulas take the form of fast converging series of powers of the log-forward moneyness and of the time-to-maturity; these series are obtained via a factorized integral representation in the Mellin space evaluated by means of residues in $\mathbb{C}$ or $\mathbb{C}^2$. Comparisons with numerical methods and efficiency tests are also discussed.
\end{abstract}

\section{Introduction}
\label{intro}

Spectrally negative L\'evy processes are L\'evy processes (see the classical textbook \cite{Bertoin96} for a complete introduction to the theory of L\'evy processes, and, among many other references, \cite{Geman02,Schoutens03,Cont04,Tankov11} for their applications in financial modeling) whose L\'evy measure is supported by the real negative axis, i.e., processes without positive jumps \cite{Kuznetsov12}; they include Brownian motion with drift, asymmetric $\alpha$-stable \cite{Zolotarev86,Carr03} or asymmetric tempered-stable \cite{Carr02,Poirot06} processes and their particular cases, such as negative Gamma and Inverse Gaussian processes. Such one-sided processes have been shown to be effective for modeling the price of financial assets, because their heavy-tail induces a leptokurtosis in the distribution of returns (whose empirical evidence is known since \cite{Fama65}), and their skewed behavior introduces the asymmetry in the occurrence of upward and downward jumps (see \cite{Carr03,Madan08,Eberlein10} for more recent discussions and justifications). Moreover, in the context of exponential market models \cite{Schoutens03,Cont04}, they generate a wide range of dynamics for the log returns, from almost surely continuous trajectories in the Brownian motion case \cite{Black73}, to highly discontinuous realizations with a potentially infinite number of downward jumps on any given time interval.

For the specific purpose of option pricing, spectrally negative L\'evy processes have been introduced in \cite{Carr03} in the case of a totally skewed $\alpha$-stable dynamics, the strong asymmetry of the model combined with the presence of fat tails capturing volatility patterns for longer observable horizons more accurately than Gaussian models. They have subsequently been employed in the calibration of index options on some major equity indices  (it is shown in \cite{Eberlein10} that positive jumps are not needed for long term options on most index markets); concerning path-dependent instruments, the impact of one-sided L\'evy dynamics on Asian and Barrier options has also been investigated \cite{Patie13,Avram02}. Let us also mention that spectrally negative L\'evy processes have been successfully applied in other areas of Quantitative Finance, notably in default modeling and credit exposure \cite{Madan08}, as the default of a firm is often linked to brutal losses in their assets' value.

When it comes to practical evaluation however, things are more complicated under L\'evy dynamics than in the usual Black-Scholes framework; the literature is dominated by numerical (finite difference) schemes for Partial Integro-Differential Equations \cite{Cont05}, by Monte Carlo simulations \cite{Poirot06} or Fourier transforms of option prices \cite{Carr99}. The latter approach is particularly popular, because in most exponential L\'evy models, the characteristic function of the asset's log price is available in a closed and relatively compact form; several refinements of the method have been introduced to accelerate the evaluation of Fourier integrals, notably by means of other integral transforms (among others, Fourier-cosine transform \cite{Fang08} or Hilbert transform \cite{Feng08}) or, more recently, by application of frame duality properties \cite{Kirkby15}. Let us also mention that, in the case of European style options with non standard terminal payoffs, it is also possible to use an integral representation decomposing the contract into a sum of standard contracts (see \cite{Carr01}) and to evaluate numerically the involved integrals, at least for smooth (twice differentiable) payoffs. For discontinuous payoffs, it is advantageous to refine this approach by assembling a series of selected payoffs and by using the theory of frames, as developed in \cite{Kirkby19}.

In this paper, we would like to take profit of the properties of another Fourier-related transform, namely the Mellin transform \cite{Flajolet95}. First, let us mention that the Mellin transform has been previously implemented in many areas of financial modeling, from providing representations for vanilla or basket options in the Black-Scholes model \cite{Panini04}, to quantifying the at-the-money implied volatility slope in various L\'evy models \cite{Gerhold16}. In our approach, we will focus on expressing Mellin integrals as a sum of residues in $\mathbb{C}$ or $\mathbb{C}^2$, so as to obtain simple series expansions for option prices. More precisely, we will show that, in the framework of exponential L\'evy models driven by spectrally negative processes, option prices have a factorized form in the Mellin space (in terms of maturity and log-forward moneyness); inverting the transform, the prices can be conveniently computed by a straightforward series of residues, allowing for a very simple and fast evaluation of the options. 

The Mellin residue technique has been used to derive fast convergent series for European options prices and Greeks, in the Black-Scholes \cite{Aguilar19} and Finite Moment Log Stable (FMLS) \cite{AK19} models; in this article, we will show that the technique successfully applies to a more general range of exotic power-related options (Digital, Log, Gap, European with cap etc.). This family of options offers a higher (and nonlinear) payoff than the vanilla options; it is used e.g. to increase the leverage ratio of strategies, or to lock the exposition to future volatility (see discussion in \cite{Tompkins98}). From a more theoretical point of view, power payoffs have also been employed to determine the L\'evy symbol of the underlying asset dynamics in \cite{Bouzianis19}. In the Gaussian context, closed formulas for pricing and hedging standard power options are known since \cite{Heynen96}, and have been recently generalized to include some barrier features \cite{Ibrahim13}; studies have also been made in the setup of local volatility models, or for more generic polynomial options (decomposed a sum of power options) in \cite{Macovschi06}. The present paper will be devoted to establishing efficient pricing formulas in the context of an asymmetric $\alpha$-stable exponential L\'evy model, and to show that it is possible to extend them to the more generic class of tempered stable processes.

The paper is organised as follows: in section~\ref{sec: Exponential Levy} we start by recalling some basic facts about option pricing in exponential L\'evy models; then, in section~\ref{sec:FMLS}, we establish a factorized form for option prices in the case of a spectrally negative $\alpha$-stable dynamics. This factorized form enables us to derive several pricing formulas for power-related instruments in section~\ref{sec:payoffs}, under the form of fast convergent series of powers of the time-to-maturity and of the moneyness; in this section, we also test the results numerically, and provide efficiency tests. In section~\ref{sec:tempered}, we show that similar formulas can also be derived if the stable distribution is tempered, and study the impact of the tempering parameter in the case of a digital option. Finally, section~\ref{sec:conclusion} is devoted to concluding remarks and perspectives.


\section{Option pricing in exponential L\'evy Models}
\label{sec: Exponential Levy}

\subsection{Model specification}
\label{subsec:Model def}

\paragraph{Notations.} Given a filtered probability space $ (\Omega, \mathcal{F}, \{ \mathcal{F}_t \}_{t \geq 0}, \mathbb{P} ) $, recall that a process $\{ X_t \}_{t \geq 0}$ is a \textit{L\'evy process} \cite{Bertoin96,Kyprianou13} if there exists a triplet $(a,b,\nu)$ such that the \textit{characteristic exponent} $\Psi(k):=-\log \mathbb{E} ^\mathbb{P} [ e^{i k X_1} ] $ of $X_t$ admits the representation
\begin{equation}\label{Levy_exponent}
    \Psi(k) \, = \, iak \, + \, \frac{1}{2}bk^2 \, + \, \int_{\mathbb{R}} \, ( 1 - e^{ikx} + ikx\mathbbm{1}_{ \{ |x| < 1  \} } ) \, \nu (\ud x)
    ,
\end{equation}
where $a,b\in\mathbb{R}$ and $\nu$ is a measure concentrated on $\mathbb{R} \backslash \{ 0 \}$ satisfying 
\begin{equation}
    \int_{\mathbb{R}} \, \mathrm{min} (1,x^2) \, \nu(\ud x) \, < \, \infty
    .
\end{equation}
Equation \eqref{Levy_exponent} is known as the \textit{L\'evy-Khintchine} formula; $a$ is the \textit{drift}, $b$ is the \textit{Brownian} (or \textit{diffusion}) \textit{coefficient} and $\nu$ is the \textit{L\'evy measure} of the process. 

If $\nu(\mathbb{R})<\infty$, one speaks of a process with \textit{finite activity} or \textit{intensity}; this corresponds to processes whose realizations have a finite number of jumps on every finite interval, like in jump-diffusion models such as the Merton model \cite{Merton76} or the Kou model \cite{Kou02}. If $\nu(\mathbb{R}) = \infty$, then one speaks of a process with \textit{infinite activity} or \textit{intensity}, and in this case an infinite number of jumps occur on every finite interval; this gives birth to a very rich dynamics and such processes do not need a Brownian component to generate complex behaviors. When furthermore $\nu(\mathbb{R}_+) = 0$ (resp. $\nu(\mathbb{R}_-) = 0$), the process is said to be \textit{spectrally negative} (resp. \textit{spectrally positive}).

As a L\'evy process has stationary independent increments, its \textit{characteristic function} can be written down as
\begin{equation}
    \mathbf{F} [X_t] (k) \, := \, \mathbb{E}^\mathbb{P} [e^{i k X_t}] \, = \, e^{-t \Psi(k)}
\end{equation}
and its \textit{moment generating function}, whenever it converges, as:
\begin{equation}
    \mathbf{M} [X_t] (p) \, := \, \mathbb{E}^\mathbb{P} [e^{p X_t}] \, = \, e^{t \phi(p)} \,\, , \,\,\,\,\,\,\, \phi(p) \, = \, - \Psi (-ip)
    .
\end{equation}
The function $\phi(p)$ is the \textit{Laplace exponent} or \textit{cumulant generating function} of the process, and its existence depends on the asymptotic behavior of the L\'evy measure; in particular, in the case of a spectrally negative process, the absence of positive fat tail ensures that $\phi(p)$ exists in the whole complex half-plane $\{ Re(p) > 0 \}$.

\paragraph{Exponential processes.} Let us now introduce the class of exponential L\'evy models, following the classical setup of \cite{Schoutens03,Cont04} for instance. Let $T>0$, and let $S_t$ denote the value of a financial asset at time $t\in[0,T]$; we assume that it can be modeled as the realization of a stochastic process $\{ S_t \}_{t\geq 0}$ on the canonical space $\Omega = \mathbb{R}_+$ equipped with its natural filtration, and that, under the \textit{risk-neutral measure} $\mathbb{Q}$, its instantaneous variations can be written down in local form as:
\begin{equation}\label{SDE_exp}
    \frac{\ud S_t}{S_t} \, = \, (r - q) \, \ud t \, + \,  \, \ud X_t
    .
\end{equation}
In the stochastic differential equation \eqref{SDE_exp}, $r\in\mathbb{R}$ is the risk-free interest rate and $q\in\mathbb{R}$ is the dividend yield, both assumed to be deterministic and continuously compounded, and $\{ X_t \}_{t\geq 0}$ is a L\'evy process; for the simplicity of notations, we will assume that $q = 0$, but all the results of the paper remain valid when replacing $r$ by $r-q$.

The solution to \eqref{SDE_exp} is the \textit{exponential L\'evy process} defined by:
\begin{equation}\label{Solution_exp}
    S_T \, = \, S_t e^{(r+\mu)\tau + X_\tau}
\end{equation}
where $\tau:=T-t$ is the horizon (or time-to-maturity), and $\mu$ is the \textit{martingale} (or \textit{convexity}) \textit{adjustment} computed in a way that the discounted stock price is a $\mathbb{Q}$-martingale, which reduces to the condition:
\begin{equation}
    \mathbb{E}_t^{\mathbb{Q}} \left[ e^{\mu\tau + X_\tau} \right] \, = \, 1
    ,
\end{equation}
or, equivalently, in terms of the Laplace exponent:
\begin{equation}\label{mu_def}
    \mu \, = \, - \phi(1) 
    .
\end{equation}

\subsection{Option pricing}
\label{subsec:Option pricing}

Let $N\in\mathbb{N}$ and $\mathcal{P}: \mathbb{R}_+^{1+N} \rightarrow \mathbb{R}$ be a non time-dependent payoff function depending on the terminal price $S_T$ and on some positive parameters $K_n$, $n= 1\dots N$:
\begin{equation}
    \mathcal{P} \, : (S_T,K_1, \dots , K_N) \, \rightarrow \, \mathcal{P} (S_T,K_1, \dots , K_N) \, := \, \mathcal{P} (S_T, \underline{K} )
    .
\end{equation}
The value at time $t$ of an option with maturity $T$ and payoff $\mathcal{P}(S_T,\underline{K})$ is equal to the risk-neutral conditional expectation of the discounted payoff:
\begin{equation}\label{Risk-neutral_1}
    C(S_t,\underline{K},r,\mu,t,T) \, = \, \mathbb{E}_t^{\mathbb{Q}} \left[e^{-r\tau} \mathcal{P} (S_T,\underline{K}) \right]
    .
\end{equation}
In the case where the L\'evy process admits a density $g(x,t)$, then, using \eqref{Solution_exp}, we can re-write \eqref{Risk-neutral_1} by integrating all possible realizations for the terminal payoff over the probability density:
\begin{equation}\label{Risk-neutral_2}
    C(S_t,\underline{K},r,\mu,\tau) \, = \, e^{-r\tau} \, \int\limits_{-\infty}^{+\infty} \, \mathcal{P} \left( S_t e^{(r+\mu)\tau +x} , \underline{K}  \right) g(x,\tau) \, \ud x
    .
\end{equation}
In all the following and to simplify the notations, we will forget the $t$ dependence in the stock price $S_t$.

\section{Spectrally negative $\alpha$-stable process (FMLS process)}
\label{sec:FMLS}

\subsection{L\'evy-stable process}
\label{subsec:Lévy-stable}

A \textit{L\'evy-stable} process \cite{Samorodnitsky94,Zolotarev86} is a L\'evy process whose L\'evy-Khintchine triplet has the form $(a,0,\nu_{stable})$, with:
\begin{equation}\label{nu_stable}
    \nu_{stable} (x) \, = \, \frac{\gamma_-}{|x|^{1+\alpha}}\mathbbm{1}_{\{x<0\}} \, + \, \frac{\gamma_+}{x^{1+\alpha}}\mathbbm{1}_{\{x>0\}}
    ,
\end{equation}
where $\alpha\in(0,2)$ and $\gamma_{\pm}\in\mathbb{R}$. It is known that, introducing $\gamma$ and $\beta$ defined by
\begin{equation}
    \left\{
    \begin{aligned}
    & \gamma^\alpha \, : = \, -(\gamma_+ + \gamma_-) \Gamma(-\alpha)\cos\frac{\pi\alpha}{2} \\
    & \beta \, := \, \frac{\gamma_+ - \gamma_-}{\gamma_+ + \gamma_-}
    ,
    \end{aligned}
    \right.
\end{equation}
then for $\alpha\in (0,1) \cup (1,2]$ the characteristic exponent of the process admits the Feller parametrization:
\begin{equation}\label{Psi_stable}
    \Psi_{stable}(k) \, = \, \gamma^\alpha |k|^\alpha \left( 1-i\beta\tan\frac{\alpha\pi}{2}\mathrm{sgn}k   \right) \, + \, i\eta k
\end{equation}
for some constant $\eta\in\mathbb{R}$ (see, for instance, exercise 1.4 in the textbook \cite{Kyprianou13}). A L\'evy-stable process can therefore be represented as a 4-parameter process $L_(\eta,\gamma^\alpha,\beta$): $\alpha$ controls the behavior of the tails and $\beta\in[-1,1]$ their \textit{asymmetry}, $\gamma$ is a \textit{scale} parameter, and $\eta$ is a \textit{location} parameter. In particular, when $\alpha\in (1,2]$ then it follows from \eqref{Psi_stable} that $\eta$ equals the mean $\mathbb{E}^{\mathbb{Q}}[X_t]$.

It is interesting to note that when $\alpha=2$ and $\eta=0$ then the characteristic function \eqref{Psi_stable} degenerates into the characteristic function of the centered normal distribution:
\begin{equation}
    L(0,\sigma^2,\beta) \, = \, N(0,(\sigma\sqrt{2})^2 ) \hspace*{0.5cm} \forall \beta\in[-1,1]
    ,
\end{equation}
and therefore the Black-Scholes model is a particular case of a L\'evy-stable model for $\alpha=2$.

\subsection{Fully asymmetric process}
\label{subsec:Fully asymmetric}

It follows from the definition of the L\'evy measure \eqref{nu_stable}, that the moment generating function $\mathbf{M} [X_t] (p)$ of a L\'evy-stable process exists if and only if $\gamma_+ =0$, or equivalently $\beta=-1$ that is, in the case of a spectrally negative L\'evy-stable process, because it has only one fat-tail located in the real negative axis; one also speaks of a \textit{fully asymmetric} process, and the condition $\beta=-1$ is known as the \textit{maximal negative asymmetry hypothesis}. In this context, choosing $\eta=0$ (process with zero mean), we have:
\begin{equation}\label{Laplace exponent}
    -\Psi_{stable}(-ip) \, = \, \gamma_-\int\limits_{-\infty}^0 (e^{px}-1) \frac{\ud x}{|x|{1+\alpha}} \, = \, \gamma_- \Gamma(-\alpha)p^\alpha \, = \, -\frac{\gamma^\alpha}{\cos\frac{\pi\alpha}{2}}p^\alpha
\end{equation}
which is valid for $p > 0$. It follows from definition \eqref{mu_def} that the martingale adjustment reads:
\begin{equation}\label{mu_FMLS}
    \mu \, = \, \frac{\gamma^\alpha}{\cos\frac{\pi\alpha}{2}}
    .
\end{equation}

It is in \cite{Carr03} that an exponential L\'evy model \eqref{SDE_exp} for a process $\{X_t\}_{t\geq 0}$ being a spectrally negative L\'evy-stable process $L(0,\sigma^\alpha,-1)$ 
was first introduced for the purpose of option pricing. The authors gave it the name of \textit{Finite Moment Log Stable (FMLS)} process, in reference to the existence of the cumulant generating function in this case. Note that the process has infinite activity, the integral of the stable measure being divergent in $0$.

\subsection{Self-similarity and option pricing}
\label{subsec:FMLS_Option pricing}

We now derive a Mellin-Barnes representation for the density of the FMLS process, and for the corresponding option price that we will denote by $C_\alpha$.

\begin{lemma}
\label{lemma:density}
Let $\sigma>0$, $\alpha\in (1,2]$ and $X_t \sim L(0,\sigma^\alpha,-1)$. 
Then the density $g_\alpha(x,t)$ of the process $\{ X_t \}_{t\geq 0 }$ admits the following Mellin-Barnes representation:
\begin{equation}\label{density_FMLS}
    g_\alpha(x,t) \, = \, \frac{1}{\alpha x } \int\limits_{c_1-i\infty}^{c_1+i\infty} \, \frac{\Gamma(1-s_1)}{\Gamma(1-\frac{s_1}{\alpha})} \, \left( \frac{x}{(-\mu t)^{\frac{1}{\alpha}}}  \right)^{s_1} \, \frac{\ud s_1}{2i\pi}
\end{equation}
where $c_1<1$ and $\mu = \frac{\sigma^\alpha}{\cos\frac{\pi\alpha}{2}}$.
\end{lemma}
\begin{proof}
Using eqs. \eqref{Laplace exponent} and \eqref{mu_FMLS} and the Laplace inversion formula, we have:
\begin{equation}
    g_\alpha(x,t) \, = \, \frac{1}{2i\pi} \, \int\limits_{c_p-i\infty}^{c_p+i\infty} e^{-px} e^{-\mu t p^\alpha} \, \ud p 
\end{equation}
where $c_p>0$. Taking the Mellin transform and making the change of variables $p^\alpha \rightarrow p$, we have:
\begin{equation}
    g_\alpha^*(s_1,t) \, := \, \int\limits_{0}^\infty g(x,t)  x^{s_1-1} \ud x \, = \,  \frac{1}{\alpha} \Gamma(s) \, \frac{1}{2i\pi} \int\limits_{c_p-i\infty}^{c_p+i\infty} e^{-\mu t p} p^{\frac{1-s}{\alpha}-1} \ud p 
\end{equation} 
for any $s>0$. The remaining $p$-integral is equal to $\frac{1}{\Gamma(1+\frac{s-1}{\alpha})}(-\mu t)^\frac{1-s}{\alpha}$ on the condition that $s>1-\alpha$ (see for instance \cite{Bateman54} or any monograph on Laplace transform); observe that, as $\alpha\in(1,2]$, the two conditions on $s$ reduce to $s>0$. Finally, the integral \eqref{density_FMLS} is obtained by applying the Mellin inversion formula \cite{Flajolet95} and by changing the variable $s\rightarrow 1-s$. \qed
\end{proof}

Equation \eqref{density_FMLS} shows that the density is a function of the ratio $\frac{x}{(-\mu t)^{\frac{1}{\alpha}}}$, which is actually a consequence of the self-similarity property \cite{Embrechts00} of stable processes (a scaling of time is equivalent to an appropriate scaling of space). This property allows for a nice factorization of the option price in the Mellin space; indeed, let us denote
\begin{equation}\label{Mellin_G}
    G_\alpha^*(s_1) \, := \, \frac{1}{\alpha}\frac{\Gamma(1-s_1)}{\Gamma(1-\frac{s_1}{\alpha})}
\end{equation}
and 
\begin{equation}\label{Mellin_K}
    K^*(s_1) \, := \, \int\limits_{-\infty}^{+\infty} \,  \mathcal{P} \left( S e^{(r+\mu)\tau +x} , \underline{K}  \right) x^{s_1-1} \, \ud x
    ,
\end{equation}
and let us assume that the integral \eqref{Mellin_K} converges for $Re(s_1)\in (c_- , c_+)$ for some real numbers $c_- < c_+$. Then, as a direct consequence of the pricing formula \eqref{Risk-neutral_2} and of lemma \ref{lemma:density}, we have:
\begin{proposition}[Factorization in the Mellin space]
\label{proposition:convolution}
Let $c_1\in(\tilde{c}_-,\tilde{c}_+)$ where
$
(\tilde{c}_-,\tilde{c}_+) := (c_-,c_+) \cap (-\infty , -1)
$ is assumed to be nonempty. Then, under the hypothesis of lemma~\ref{lemma:density}, the value at time $t$ of an option with maturity $T$ and payoff $\mathcal{P}(S_T,\underline{K})$ is equal to:
\begin{equation}\label{Risk-neutral:convolution}
    C_\alpha (S,\underline{K},r,\mu,\tau) \, = \, e^{-r\tau}  \int\limits_{c_1-i\infty}^{c_1+i\infty} \, K^*(s_1) G_\alpha^*(s_1) \, \left( -\mu\tau \right)^{-\frac{s_1}{\alpha}} \, \frac{\ud s_1}{2i\pi}
    .
\end{equation}
\end{proposition}

The factorized form \eqref{Risk-neutral:convolution} turns out to be a very practical tool for option pricing. Indeed, as an integral along a vertical line in the complex plane, it can be conveniently expressed as a sum of residues associated to the singularities of the integrand. As Gamma functions are involved, we can control the behavior of the integrand when the contour goes to infinity by using the Stirling asymptotic formula for the Gamma function \cite{Abramowitz72}: if $a_k$, $b_k$, $c_j$, $d_j$ are real numbers, if $\delta:=\sum_k a_k - \sum_j c_j$ and if $\delta < 0$ then
\begin{equation}\label{Stirling}
    \left| \frac{\Pi_k \Gamma(a_k s+b_k)}{\Pi_j \Gamma(c_j s+d_j)} \right| 
\overset{|s|\rightarrow \infty}{\longrightarrow} 0
\end{equation}
when $\arg s \, \in \, (-\frac{\pi}{2} , \frac{\pi}{2})$, and the same holds for $\arg s \, \in \, (\frac{\pi}{2} , \frac{3\pi}{2})$ if $\delta>0$. Therefore, by right or left closing the contour of integration in \eqref{Risk-neutral:convolution}, the option price will take the form of a series:
\begin{equation}
    e^{-r\tau} \times \sum \, \left[ \textrm{residues of }K^*(s_1)G_\alpha^*(s_1) \times \textrm{powers of } \left(-\mu\tau\right)^\frac{1}{\alpha} \right]
    .
\end{equation}
The only technical difficulty will in fact lie in the evaluation of $K^*(s_1)$: depending on the payoff's complexity, it can be either computed directly, or via the introduction of a second Mellin complex variable $s_2$. 

\section{Power payoffs in a spectrally negative $\alpha$-stable environment}
\label{sec:payoffs}

In all this section, $\alpha\in (1,2]$, $\sigma>0$, $X_t \sim L(0,\sigma^\alpha,-1)$ and $u>0$; the \textit{log-forward} moneyness is defined to be:
\begin{equation}\label{moneyness}
   k_u \, := \, \log\frac{S}{K^{\frac{1}{u}}} + r\tau
\end{equation}
and we will use the standard notation $X^+:= X \mathbbm{1}_{ \{X>0 \}}$.

\subsection{One complex variable payoffs}
\label{subsec:1-complex}

\paragraph{Digital power options (cash-or-nothing).}
The call's payoff is:
\begin{equation}
    \mathcal{P}^{(C/N)} (S,K) \, := \, \mathbbm{1}_{ \{ S^u - K > 0 \}  }
    .
\end{equation}

\begin{proposition}
\label{prop:C/N}
The value at time $t$ of a digital power cash-or-nothing call option is:
\begin{equation}\label{Call_C/N}
    C_\alpha^{(C/N)} (S,K,r,\mu,\tau) \, = \, \frac{e^{-r\tau}}{\alpha} \, \sum\limits_{n=0}^\infty \, \frac{1}{n!\Gamma\left( 1 - \frac{n}{\alpha} \right)} \, (k_u + \mu\tau)^n (-\mu\tau)^{-\frac{n}{\alpha}}
    .
\end{equation}
\end{proposition}
\begin{proof}
As we can write:
\begin{eqnarray}
    \mathcal{P}^{(C/N)} (Se^{(r+\mu)\tau+x},K) & = &   \mathbbm{1}_{ \{ e^{u (k_u+ \mu\tau +x) } > 1 \} }  \nonumber  \\
   & = & \mathbbm{1}_{ \{ x  > -k_u - \mu\tau \}  } 
   ,
\end{eqnarray}
then, with the notation \eqref{Mellin_K}, the $K^*(s_1)$ function reads:
\begin{equation}\label{Mellin_K_DigitalCN}
K^*(s_1) \, = \, -\frac{(-k_u-\mu\tau)^{s_1}}{s_1} 
\end{equation}
for $s_1<-1$. Using proposition \ref{proposition:convolution} and the functional relation $\Gamma(-s_1)=-s_1\Gamma(1-s_1)$, the option price is:
\begin{equation}\label{Call_C/N_MB}
    C_\alpha^{(C/N)} (S,K,r,\mu,\tau) \, = \, \frac{e^{-r\tau}}{\alpha} \, \int\limits_{c_1-i\infty}^{c_1+i\infty} \, \frac{\Gamma(-s_1)}{\Gamma(1-\frac{s_1}{\alpha})} (-k_u-\mu_\tau)^{s_1} (-\mu\tau)^{-\frac{s_1}{\alpha}} \, \frac{\ud s_1}{2i\pi}
\end{equation}
which converges for $s_1<0$. We can note that:
\begin{equation}
    \delta \, = \, \frac{1}{\alpha} - 1
\end{equation}
is negative because $\alpha>1$, thus, it follows from the Stirling formula \eqref{Stirling} that the analytic continuation of the integrand vanishes at infinity in the right half plane. Therefore, the integral \eqref{Call_C/N_MB} equals the sum of residues at the poles located in this half plane; these poles are induced by the $\Gamma(-s_1)$ term at every positive integer $n$, and the associated residues are:
\begin{equation}
    \frac{(-1)^n}{n!}\frac{1}{\Gamma(1-\frac{n}{\alpha})} (-k_u-\mu\tau)^n (-\mu\tau)^{-\frac{n}{\alpha}}
    .
\end{equation}
Simplifying and summing all residues yields \eqref{Call_C/N}.\qed
\end{proof}

\paragraph{Log power options.}
These options were introduced in \cite{Wilmott06} in the case $u=1$, and are basically options on the rate of return of the underlying asset. The call's payoff is:
\begin{equation}
    \mathcal{P}^{(Log)} (S,K) \, := 
    \left[  \log\left( \frac{S^u}{K} \right)  \right]^+
\end{equation}

\begin{proposition}
\label{prop:Log}
The value at time $t$ of a Log power call option is:
\begin{equation}\label{Call_Log}
    C_\alpha^{(Log)} (S,K,r,\mu,\tau) \, = \, \frac{u e^{-r\tau}}{\alpha} \, \sum\limits_{n=0}^\infty \, \frac{1}{n!\Gamma\left( 1 + \frac{1-n}{\alpha} \right)} \, (k_u + \mu\tau)^n (-\mu\tau)^{\frac{1-n}{\alpha}}
    .
\end{equation}
\end{proposition}
\begin{proof}
As we can write:
\begin{equation}
    \mathcal{P}^{(Log)} (Se^{(r+\mu)\tau+x},K) \, = \, u \, [ k_u + \mu\tau +x  ]^+
    ,
\end{equation}
then the $K^*(s_1)$ function reads:
\begin{equation}
K^*(s_1) \, = \, u \, \frac{(-k_u-\mu\tau)^{1+s_1}}{s_1(1+s_1)} 
\end{equation}
for $s_1<-1$. Using proposition \ref{proposition:convolution} and the functional relation $\Gamma(1-s_1)=-s_1\Gamma(-s_1)$, the option price is:
\begin{multline}
    C_\alpha^{(Log)} (S,K,r,\mu,\tau)  \, = \\
     \frac{u e^{-r\tau}}{\alpha} \, \int\limits_{c_1-i\infty}^{c_1+i\infty} \, -\frac{\Gamma(-s_1)}{(1+s_1) \Gamma(1-\frac{s_1}{\alpha})} (-k_u-\mu_\tau)^{1+s_1} (-\mu\tau)^{-\frac{s_1}{\alpha}} \, \frac{\ud s_1}{2i\pi}
\end{multline}
which converges for $s_1<-1$. Again, $\delta<0$, and the analytic continuation of the integrand in the right half-plane has:
\begin{itemize}
    \item a simple pole in $s_1=-1$ with residue 
        \begin{equation}\label{Log_series1}
        \frac{(-\mu\tau)^{\frac{1}{\alpha}}}{\Gamma(1+\frac{1}{\alpha})}
        ;
        \end{equation}
    \item a series of poles at every positive integer $s_1=n$ with residues:
    \begin{equation}\label{Log_series2}
        -\frac{(-1)^n}{(n+1)!}\frac{1}{\Gamma(1-\frac{n}{\alpha})} (-k_u-\mu\tau)^{1+n} (-\mu\tau)^{-\frac{n}{\alpha}}
        .
    \end{equation}
\end{itemize}
Summing the residues \eqref{Log_series1} and \eqref{Log_series2} for all $n$ and re-ordering yields \eqref{Call_Log}.\qed
\end{proof}

\paragraph{Capped power options (cash-or-nothing).}
For $K_- < K_+$, the call's payoff is:
\begin{equation}
    \mathcal{P}^{(C/N,cap)} (S,K_+,K_-) \, := \,
    \mathbbm{1}_{ \{ K_- < S^u < K_+ \} } 
    .
\end{equation}
Let us define $k_u^{\pm} \, := \, \log\frac{S}{K_\pm^{\frac{1}{u}}} + r\tau$. We have:

\begin{proposition}
\label{prop:C/N_Cap}
The value at time $t$ of a capped cash-or-nothing call option is:
\begin{multline}\label{Call_C/N_Cap}
    C_\alpha^{(C/N,cap)} (S,K_+,K_-,r,\mu,\tau) \, = 
    \\
    \frac{ e^{-r\tau}}{\alpha} \, \sum\limits_{n=0}^\infty \, \frac{1}{n!\Gamma\left( 1 - \frac{n}{\alpha} \right)} \, \left( (k_u^- + \mu\tau)^n - (k_u^+ + \mu\tau)^n \right) (-\mu\tau)^{-\frac{n}{\alpha}}
    .
\end{multline}
\end{proposition}
\begin{proof}
We can write:
\begin{equation}\label{payoff_cap_identity}
    \mathcal{P}^{(C/N,cap)} (Se^{(r+\mu)\tau+x},K_+,K_-) \, = \, \mathbbm{1}_{ \{ -k_u^- - \mu\tau < x < -k_u^+ - \mu\tau  \} }
\end{equation}
and therefore the $K^*(s_1)$ function reads:
\begin{equation}
K^*(s_1) \, =  \, \frac{(-k_u^+ -\mu\tau)^{s_1} - (-k_u^- -\mu\tau)^{s_1}}{s_1} 
.
\end{equation}
From proposition~\ref{proposition:convolution}, the option price is:
\begin{multline}
    C_\alpha^{(C/N, cap)} (S,K_+,K_-,r,\mu,\tau)  \, = \\
     \frac{ e^{-r\tau}}{\alpha} \, \int\limits_{c_1-i\infty}^{c_1+i\infty} \, -\frac{\Gamma(-s_1)}{ \Gamma(1-\frac{s_1}{\alpha})} 
     ( (-k_u^+ -\mu\tau)^{s_1} - (-k_u^- -\mu\tau)^{s_1} )
     (-\mu\tau)^{-\frac{s_1}{\alpha}} \, \frac{\ud s_1}{2i\pi}
     .
\end{multline}
Like in proposition~\ref{prop:C/N}, summing all the residues associated to the poles of the $\Gamma(-s_1)$ function yields the series \eqref{Call_C/N_Cap}.\qed
\end{proof}

\subsection{Two complex variables payoffs}
\label{subsec:2-complex}

\paragraph{Digital power options (asset-or-nothing).}
The call's payoff is:
\begin{equation}
    \mathcal{P}^{(A/N)} (S,K) \, := \, S^u \mathbbm{1}_{ \{ S^u - K > 0 \}  }
\end{equation}

\begin{proposition}
\label{prop:A/N}
The value at time $t$ of a digital power asset-or-nothing call option is:
\begin{equation}\label{Call_A/N}
    C_\alpha^{(A/N)} (S,K,r,\mu,\tau) \, = \, \frac{K e^{-r\tau}}{\alpha} \, \sum\limits_{\substack{n = 0 \\ m = 0}}^\infty \, \frac{1}{n!\Gamma\left( 1 + \frac{m-n}{\alpha} \right)} \, u^m (k_u + \mu\tau)^n (-\mu\tau)^{\frac{m-n}{\alpha}}
    .
\end{equation}
\end{proposition}
\begin{proof}
We can write:
\begin{equation}
    \mathcal{P}^{(A/N)} (Se^{(r+\mu)\tau+x},K) \, = \, Ke^{u(k_u+\mu\tau +x)} \, \mathbbm{1}_{ \{ x  > -k_u - \mu\tau \} }
    .
\end{equation}
Introducing a Mellin-Barnes representation for the exponential term:
\begin{equation}
    e^{u(k_u+\mu\tau +x)} = \int\limits_{c_2-i\infty}^{c_2+i\infty} (-1)^{-s_2} u^{-s_2} \Gamma(s_2) (k_u+\mu\tau +x)^{-s_2} \, \frac{\ud s_2}{2i\pi}
\end{equation}
for $c_2>0$ and integrating over the $x$ variable, the $K^*(s_1)$ function reads:
\begin{multline}\label{Mellin_K_DigitalAN}
K^*(s_1) \, = 
\\
K \int\limits_{c_2-i\infty}^{c_2+i\infty} (-1)^{-s_2} u^{-s_2} \frac{\Gamma(s_2)\Gamma(1-s_2)\Gamma(-s_1+s_2)}{\Gamma(1-s_1)} (-k_u-\mu\tau)^{s_1-s_2} \frac{\ud s_2}{2i\pi}
\end{multline}
and converges for $(s_1,s_2)$ in the triangle $\{Re(s_2)\in(0,1),Re(s_1)<Re(s_2)\}$. From proposition~\ref{proposition:convolution}, the option price is:
\begin{multline}\label{Call_A/N-2}
    C_\alpha^{(A/N)} (S,K,r,\mu,\tau)  \, = \, \frac{K e^{-r\tau}}{\alpha}  \int\limits_{c_1-i\infty}^{c_1+i\infty} 
     \int\limits_{c_2-i\infty}^{c_2+i\infty} \\
     (-1)^{-s_2}  \frac{\Gamma(s_2)\Gamma(1-s_2)\Gamma(-s_1+s_2)}{\Gamma(1-\frac{s_1}{\alpha})} 
     u^{-s_2} (-k_u-\mu\tau)^{s_1-s_2} (-\mu\tau)^{-\frac{s_1}{\alpha}} \, \frac{\ud s_1\ud s_2}{(2i\pi)^2}
     .
\end{multline}
Poles of the integrand occur when $\Gamma(s_2)$ and $\Gamma(-s_1+s_2)$ are singular; performing the change of variables $-s_1+s_2\rightarrow U$, $s_2\rightarrow V$ allows to compute the associated residues, which read:
\begin{equation}
    (-1)^m\frac{(-1)^n}{n!}\frac{(-1)^m}{m!}\frac{\Gamma(1+m)}{\Gamma(1+\frac{m-n}{\alpha})} u^m (-k_u - \mu\tau)^n (-\mu\tau)^{\frac{m-n}{\alpha}}
\end{equation}
Simplifying and summing the residues yields the series \eqref{Call_A/N}. The fact that one can close the $\mathbb{C}^2$ contour in \eqref{Call_A/N-2} is a consequence of the multidimensional generalization of the Stirling estimate \eqref{Stirling} (see \cite{Passare97} or the appendix of \cite{Aguilar19} for details).\qed
\end{proof}

\paragraph{Gap power options.}
A gap option \cite{Tankov10} 
offers a nonzero payoff on the condition that a trigger price is attained at $t=T$. More precisely, the call's payoff is:
\begin{equation}\label{Payoff_Gap}
    \mathcal{P}^{(Gap)} (S,K_1,K_2) \, := \, (S^u-K_1) \mathbbm{1}_{ \{ S^u - K_2 > 0 \}  }
\end{equation}
where $K_1$ is the strike price and $K_2$ the trigger price; if the trigger is lower than the strike then a negative payoff is possible (which would not be the case with a classical knock-in barrier). From the definition of the payoff \eqref{Payoff_Gap}, the value of the gap call option is equal to:
\begin{equation}\label{Call_gap}
    C_\alpha^{(Gap)} (S,K_1,K_2,r,\mu,\tau) \, = \,  C_\alpha^{(A/N)} (S,K_2,r,\mu,\tau) \, - \, K_1 C_\alpha^{(C/N)} (S,K_2,r,\mu,\tau)
    .
\end{equation}

\paragraph{European power options.} The classical European power option is a gap power option with equal strike and trigger prices ($K_1=K_2=K$); the payoff therefore reads
\begin{equation}
    \mathcal{P}^{(E)} (S,K) \, := \, [S^u - K]^+ 
    .
\end{equation}
Observing that \eqref{Call_C/N} is actually a particular case of \eqref{Call_A/N} for $m=0$, it follows immediately from \eqref{Call_gap} that the value of the European power call is:
\begin{equation}\label{Call_European}
    C_\alpha^{(E)} (S,K,r,\mu,\tau) \, = \, \frac{K e^{-r\tau}}{\alpha} \, \sum\limits_{\substack{n = 0 \\ m = 1}}^\infty \, \frac{1}{n!\Gamma\left( 1 + \frac{m-n}{\alpha} \right)} \, u^m (k_u + \mu\tau)^n (-\mu\tau)^{\frac{m-n}{\alpha}}
    .
\end{equation}
When the asset is at-the-money (ATM) forward, that is when $S=K^{\frac{1}{u}}e^{-rt}$, or, equivalently, $k_u=0$, then \eqref{Call_European} becomes:
\begin{multline}\label{Call_European_ATM}
    C_\alpha^{(E,ATM)} (S,K,r,\mu,\tau) \, = \, \\
    \frac{K e^{-r\tau}}{\alpha} \, \left[ u\frac{(-\mu\tau)^{\frac{1}{\alpha}}}{\Gamma(1+\frac{1}{\alpha})}  - u(-\mu\tau) + u^2 \frac{(-\mu\tau)^{\frac{2}{\alpha}}}{\Gamma(1+\frac{2}{\alpha})} \, + \, O \left( u^2 (-\mu\tau)^{1+\frac{1}{\alpha}} \right) \right]
    .
\end{multline}
In particular, if we choose $\alpha=2$ and the normalization $\gamma=\frac{\sigma}{\sqrt{2}}$ in the definition of the martingale adjustment \eqref{mu_FMLS}, then \eqref{Call_European_ATM} reads:
\begin{eqnarray}
    & C_2^{(E,ATM)} (S,K,r,\sigma,\tau) & = \frac{K e^{-r\tau}}{2} \, \left[ 2u\frac{\sigma\sqrt{\tau}}{\sqrt{2\pi}} -u(1-u)\frac{\sigma^2}{2}\tau \, + \,  O \left( u^2 (\sigma\sqrt{\tau})^{3} \right) \right] \nonumber 
    \\
    & & \underset{u\rightarrow 1}{=} \frac{1}{\sqrt{2\pi}} S \sigma\sqrt{\tau} \, + \, O \left( (\sigma\sqrt{\tau})^{3} \right) 
\end{eqnarray}
which is the well-known approximation for the ATM Black-Scholes call.

\paragraph{Capped power options (asset-or-nothing, European).}
For $K_- < K_+$, the payoff of a capped power asset-or-nothing call is:
\begin{equation}
    \mathcal{P}^{(A/N,cap)} (S,K_+,K_-) \, := \,
    S^u \, \mathbbm{1}_{ \{ K_- < S^u < K_+ \} } 
    .
\end{equation}
The presence of a cap allows the seller to protect themselves against the eventuality of enormous payoffs; using the identity \eqref{payoff_cap_identity} for the indicator function, and proceeding in a similar way than for proving proposition~\ref{prop:A/N}, we obtain:
\begin{proposition}
\label{prop:A/N_cap}
The value at time $t$ of a capped power asset-or-nothing call option is:
\begin{multline}\label{Call_A/N_cap}
    C_\alpha^{(A/N, cap)} (S,K_+,K_-,r,\mu,\tau)  \, = \,  \frac{ e^{-r\tau}}{\alpha} S^u e^{u(r+\mu)\tau} \, \times
    \\
    \sum\limits_{\substack{n=0 \\ m=0}}^\infty \, \frac{(-u)^m \left( (k_u^- + \mu\tau)^{1+n+m} - (k_u^+ + \mu\tau)^{1+n+m} \right)}{(1+n+m) n!m! \Gamma\left( 1 - \frac{1+n}{\alpha} \right)} \, (-\mu\tau)^{-\frac{1+n}{\alpha}}
    .
\end{multline}
\end{proposition}
The value of the capped European power option is easily deduced from the values of the capped cash-or-nothing \eqref{Call_C/N_Cap} and asset-or-nothing \eqref{Call_A/N_cap} options:
\begin{multline}\label{Call_European_cap}
    C_\alpha^{(E/N, cap)} (S,K_+,K_-,r,\mu,\tau)  \, = \\
     C_\alpha^{(A/N, cap)} (S,K_+,K_-,r,\mu,\tau) \, - \, K_- C_\alpha^{(C/N, cap)} (S,K_+,K_-,r,\mu,\tau)
     .
\end{multline}
When $K_+ \rightarrow \infty$, the value of the capped option \eqref{Call_European_cap} coincides with the classical uncapped option \eqref{Call_European}; this situation is displayed in figure~\ref{fig:Cap}. We can observe that the convergence to the uncapped price is quicker when $\alpha$ decreases, which is no surprise given the overall $\frac{1}{\alpha}$ factor.
 \begin{figure}[h]
 \centering
  \includegraphics[scale=0.25]{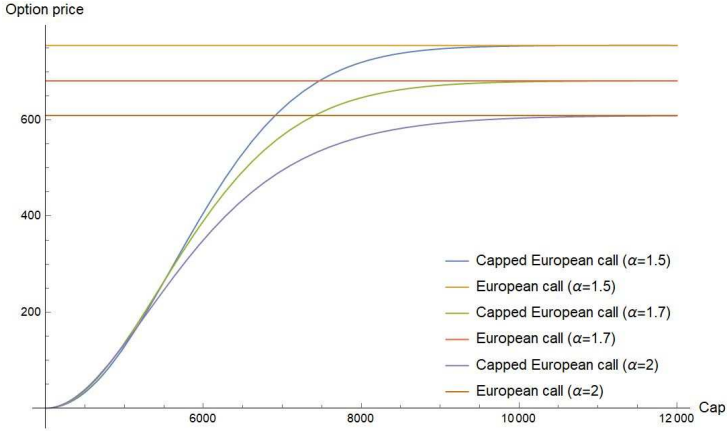}
 \caption{Convergence of capped European call to the uncapped price when the cap $K_+$ goes to infinity, for different tail index parameters $\alpha$; when $\alpha$ decreases, the European (uncapped) price grows higher, given the presence of a left fat tail as soon as $\alpha < 2$. Parameters: strike  $K_-=4000$ and horizon $\tau=2Y$; market parameters are set to $S=4200$, $r=1\%$ and $\sigma=1\%$.}
 \label{fig:Cap}      
 \end{figure}

\subsection{Numerical tests}
\label{subsec:Numerical_tests}
In this subsection, we benchmark the pricing formulas established in the previous sections by comparing them with the formulas available in the cases $\alpha=2$, $u=1$ (i.e., in the Black-Scholes setup); we also provide comparisons with numerical evaluation of Fourier integrals when $\alpha\neq 2$. Except otherwise stated, we choose $r=1\%$, $\sigma=20\%$, $K=4000$, $\tau= 2$ years and we make the normalization $\gamma=\frac{\sigma}{\sqrt{2}}$ in the martingale adjustment \eqref{mu_FMLS}, so as to recover the Black-Scholes adjustment $-\frac{\sigma^2}{2}$ when $\alpha=2$.

\paragraph{Log options}
When $\alpha=2$ and $u=1$, a closed pricing formula exists for the Log option \cite{Haug07}:
\begin{equation}\label{Log_BS}
    C_2^{(Log)} (S,K,r,\sigma,\tau) \, = \, e^{-r\tau} \sigma\sqrt{\tau}
    \left[  n(d_2)  + d_2 N(d_2) \right] 
    \,\, , \,\,\,\,\, d2:=\frac{k-\frac{\sigma^2}{2}\tau}{\sigma\sqrt{\tau}}
    ,
\end{equation}
where $k:=k_1$, $n(x)=\frac{1}{\sqrt{2 \pi}}e^{-\frac{x^2}{2}}$ is the Gaussian density and $N(x)$ is the Normal cumulative distribution function. In table~\ref{tab:Log}, we compare this formula to various truncations of the series \eqref{Call_Log} for $\alpha=2$ and $u=1$, in several market situations (out-of-the money, at-the-money and in-the-money).

 \begin{table}[h]
 \caption{Log call ($\mathcal{P}^{(Log)} (S,K) = 
  \left[ \log S^u - \log K \right]^+$): comparisons between the series \eqref{Call_Log} truncated at $n=n_{max}$ and the closed formula \eqref{Log_BS} in the case $\alpha=2$, $u=1$. We observe that very few terms are needed to obtain an excellent degree of precision, even in deeply out or in the money situations.}
 \label{tab:Log}       
 \centering
 \begin{tabular}{ccccc}
 \hline\noalign{\smallskip}
 & $n_{max}=3$ & $n_{max}=5$ & $n_{max}=10$ & Formula \eqref{Log_BS}  \\
 \noalign{\smallskip}\hline\noalign{\smallskip}
$S=5000$  & 0.238691 & 0.237465 & 0.237525 & 0.237525  \\
$S=4200$ & 0.125287 & 0.125286 & 0.125286 & 0.125286  \\
ATM & 0.092106 & 0.092104 & 0.092104 & 0.092104  \\
$S=3800$ & 0.079177 & 0.079158 & 0.079158& 0.079158  \\
$S=3000$ & 0.025250 & 0.018797& 0.019488 & 0.019487  \\
 \noalign{\smallskip}\hline
\end{tabular}
 \end{table}

\paragraph{Power options ($\alpha=2$)}
For $u>0$, recall the formula by Heynen and Kat \cite{Heynen96} for European power options in the Black-Scholes setup:
\begin{equation}\label{Heynen}
    C_2^{(E)} (S,K,r,\sigma,\tau) \, = \, S^u e^{(u-1)(r+u\frac{\sigma^2}{2})\tau}N(d_1) - Ke^{-r\tau}N(d_2)
    ,
\end{equation}
where 
\begin{equation}
    d_1 := \frac{k_u + (u-\frac{1}{2})\sigma^2\tau}{\sigma\sqrt{\tau}} \,\, , \,\,\,\,\, d_2 \, := \, d_1 - u\sigma\sqrt{\tau}
    .
\end{equation}
In table~\ref{tab:Power}, values obtained with formula \eqref{Heynen} are compared to various truncations of the series \eqref{Call_European}, for various powers $u>0$ and in the ATM situation. The convergence is very fast; of course if one is far from the money, the convergence becomes slightly slower because the moneyness $k_u$ grows when $u\neq 1$ (for instance, if $S=4500$, $k_1=0.14$ but $k_{1.5}=2.90$ and $k_3=5.67$), and therefore the powers $(k_\mu+\tau)^n$ in the numerator are less quickly neutralized by the factorial/Gamma terms of the denominator.

 \begin{table}[h]
 \caption{European power call ($\mathcal{P}^{(E)}(S,K)=[S^u-K]^+$): comparisons between the series \eqref{Call_European} for the  truncated at $n_{max}=m_{max}:=max$, and the values obtained via the formula \eqref{Heynen}, for $\alpha=2$ and various positive powers $u$.}
 \label{tab:Power}       
 \centering
 \begin{tabular}{lccccc}
 \hline\noalign{\smallskip}
 & ${max}= 3 $ & ${max}=5 $ & ${max}=10$ &  Heynen \& Kat \eqref{Log_BS}  \\
 \noalign{\smallskip}\hline\noalign{\smallskip}
$u=1$     & 439.65  &  440.93  &  440.94   & 440.94  \\
$u=1.5$   & 723.00 &  729.86 &  730.06     & 730.06 \\
$u=2$     & 1057.71  & 1080.49  & 1081.64  & 1081.64 \\
$u=3$     & 1908.17  & 2034.41  & 2049.37  & 2049.39 \\
 \noalign{\smallskip}\hline
\end{tabular}
 \end{table}

\paragraph{European options ($\alpha\neq 2$) }

As a consequence of the Gil-Pelaez inversion formula for the characteristic functions, the price of an European call can be decomposed into a sum of Arrow-Debreu securities of the form (see details e.g. in \cite{Bakshi00}):
\begin{equation}\label{Gil_Pelaez}
    C_\alpha^{(E)}(S,K,r,\mu,\tau) \, = \, S\Pi_1 \, - \, Ke^{-r\tau}\Pi_2
    .
\end{equation}
The price of each security can be expressed in terms of the stable characteristic function and of the log-forward moneyness \cite{Lewis01}:
\begin{equation}\label{Pi1}
    \Pi_1 \, = \, \frac{1}{2} \,+ \, \frac{1}{\pi}\int\limits_0^{\infty} \, Re \, \left[ \frac{e^{i u k} \Phi(u-i,\tau)}{iu} \right] \ud u
\end{equation}
and 
\begin{equation}\label{Pi2}
    \Pi_2 \, = \, \frac{1}{2} \,+ \, \frac{1}{\pi}\int\limits_0^{\infty} \, Re \, \left[ \frac{e^{i u k} \Phi(u,\tau)}{iu} \right] \ud u
    ,
\end{equation}
where $\Phi(u,t)$ is the risk-neutral characteristic function
\begin{equation}
    \Phi(u,t) \, := \, e^{i \mu u t}e^{-t \Psi_{stable}(u)}  
\end{equation} 
satisfying the martingale condition $\Phi(-i,t) = 1$. Given the simple form of the stable characteristic exponent \eqref{Psi_stable}, the integrals in \eqref{Pi1} and \eqref{Pi2} can be carried out very easily via a classical recursive algorithm on the truncated integration region (typically,  $u\in[0,1000]$ is sufficient for a precision goal of $10^{-8}$). In table~\ref{tab:European}, we compare the values obtained with this method with several truncations of the series \eqref{Call_European}, for a tail-index $\alpha=1.7$.

\begin{table}[h]
 \caption{European call ($\mathcal{P}^{(E)}(S,K)=[S^u-K]^+$): comparisons between the series \eqref{Call_European} truncated at $n_{max}=m_{max}:=max$, and the values obtained by the Gil-Pelaez method \eqref{Gil_Pelaez}, in the case $\alpha=1.7$, $u=1$. The convergence is very fast, in particular for ITM long term options.}
 \label{tab:European}       
 \centering
 \begin{tabular}{cccccc}
 \hline\noalign{\smallskip}
 & ${max}=3$ & ${max}=10$ & ${max}=20$ & ${max}=30$ & Gil-Pelaez \eqref{Gil_Pelaez}  \\
 \noalign{\smallskip}\hline\noalign{\smallskip}
 \multicolumn{2}{l}{{\bfseries Long term options} ($\tau=2$)} & & & \\
$S=5000$  & 1302.92 & 1309.86 & 1309.86 & 1309.86 & 1309.86   \\
$S=4200$  & 679.32 & 681.56 & 681.56 & 681.56 & 681.56  \\
ATM       & 496.87 & 498.07 & 498.07 & 498.07  & 498.07  \\
$S=3800$  & 425.76 & 426.44 & 426.44 & 426.44 & 426.44 \\
$S=3000$  & 128.50 & 92.46 & 96.50 & 96.50 & 96.50 \\
  \noalign{\smallskip}\hline\noalign{\smallskip}
 \multicolumn{2}{l}{{\bfseries Short term options} ($\tau=0.5$)} & & & \\
$S=5000$  & 1089.70  & 1075.64  & 1075.63 & 1075.63 & 1075.63 \\
$S=4200$  & 383.17 & 383.30 & 383.30 & 383.30 & 383.30 \\
ATM       & 230.47 & 203.49 & 203.49 & 203.49 & 203.49 \\
$S=3800$  & 143.53 & 143.09 & 143.09 & 143.09 & 143.09 \\
$S=3000$  & 211.44 & -27.24 & 1.04 & 1.39 & 1.39 \\
 \noalign{\smallskip}\hline
\end{tabular}
\end{table}

Like before, the convergence is very fast, and goes even faster in the ITM region; this is because the log-forward moneyness \eqref{moneyness} is positive in this zone, and therefore, as $\mu<0$, $(k_\mu+\mu\tau)$ is closer to $0$ than in the OTM zone, which accelerates the convergence of the series \eqref{Call_European}. This situation is displayed in figure~\ref{fig:European_Convergence}.

 \begin{figure}
  \includegraphics[scale=0.26]{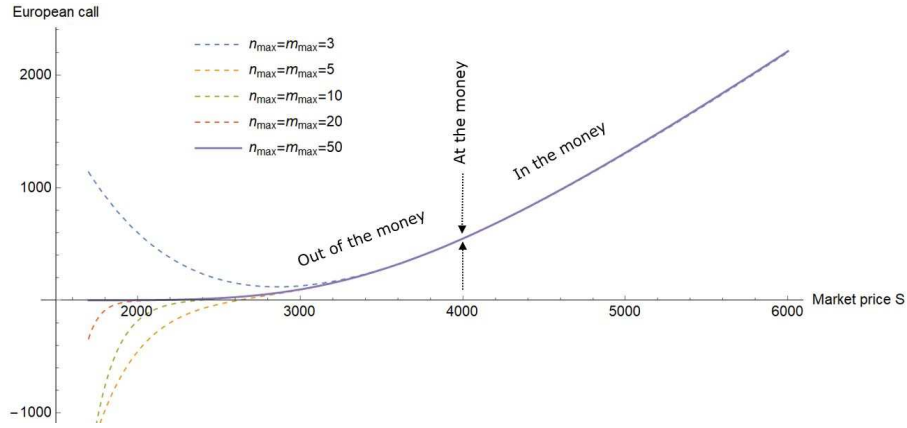}
 \caption{Convergence of partial sums of the series \eqref{Call_European} to the option price ($\alpha=1.7$); in a wide interval of prices around the money ($S\in (3000,6000)$), it is sufficient to consider only the terms up to $n_{max}=m_{max}=5$ to obtain an excellent level of precision.}
 \label{fig:European_Convergence}       
 \end{figure}
%

%

\section{Extension to one-sided tempered stable processes}
\label{sec:tempered}

Tempered stable L\'evy processes, which are known in Physics as \textit{truncated L\'evy flights}, combine $\alpha$-stable and Gaussian trends, and are an alternative solution to achieve finite moments (see details and further references in \cite{Rosinski07}). Their L\'evy-Khintchine triplet has the form $(a,0,\nu_{TS})$ where
\begin{equation}\label{nu_Temperedstable}
    \nu_{TS} (x) \, = \, \frac{\gamma_- \, e^{-\lambda_- |x|}}{|x|^{1+\alpha}}\mathbbm{1}_{\{x<0\}} \, + \, \frac{\gamma_+ \, e^{-\lambda_+ x}}{x^{1+\alpha}}\mathbbm{1}_{\{x>0\}}
\end{equation}
for $\gamma_\pm$, $\lambda_\pm \geq 0$ and $0 < \alpha_\pm < 2$. When $\gamma_-=\gamma_+$ and $\alpha_- = \alpha_+$, we recover the \textit{CGMY process} \cite{Carr02} (sometimes named \textit{classical tempered stable process}) and when furthermore $\alpha_- = \alpha_+ = 0$, the \textit{Variance Gamma process} \cite{Madan98}. In the case where $\lambda_\pm =0$, there is no more tempering and the process is simply a L\'evy-stable process like in section~\ref{sec:FMLS}. When $\alpha_\pm\in (0,1)\cup(1,2)$, the Laplace exponent of the tempered stable process can be easily computed: for $p\in (-\lambda_-,\lambda_+)$ one has
\begin{equation}\label{Laplace_exponent_TS}
    \phi(p) \, = \, \eta p \, + \, \gamma_- \Gamma(-\alpha_-) ( -\lambda_-^{\alpha_-} + (\lambda_- + p)^{\alpha_-} ) \, + \,  \gamma_+ \Gamma(-\alpha_+) ( -\lambda_+^{\alpha_+} + (\lambda_+ - p)^{\alpha_+} )
\end{equation}
where $\eta$ is a constant depending on the drift $a$ and the choice of truncation function for the characteristic function of the process; without loss of generality we choose it to be equal to $0$.

\subsection{Tempered stable densities}
\label{subsec:tempered_densities}

Let us denote by $\nu_{TS}^-(x)$ (resp. $\nu_{TS}^+(x)$) the negative (resp. positive) part of the L\'evy measure \eqref{nu_Temperedstable}, and by $TS^{\pm}(\gamma_\pm,\lambda_\pm,\alpha_\pm)$ the associated one-sided tempered stable processes.

\begin{lemma}
\label{lemma:density_TS}
Let $\alpha_\pm\in(1,2)$ and $\mu_\pm:=-\gamma_\pm\Gamma(-\alpha_\pm)$.
\begin{itemize}
    \item[(i)] If $X_t\sim TS^-(\gamma_-,\lambda_-,\alpha_-)$, then its density $g^-(x,t)$ admits the Mellin-Barnes representation:
    \begin{equation}\label{density_TS-}
        g^-(x,t) =  \frac{e^{ \lambda_-^{\alpha_-} \mu_- t \, + \, \lambda_- x}}{\alpha_- x } \int\limits_{c_1-i\infty}^{c_1+i\infty} \, \frac{\Gamma(1-s_1)}{\Gamma(1-\frac{s_1}{\alpha_-})} \, \left( \frac{x}{(-\mu_-t)^{\frac{1}{\alpha_-}}}  \right)^{s_1} \, \frac{\ud s_1}{2i\pi}
    \end{equation}
    for any $c_1 < 1$;
    \item[(ii)] If $X_t\sim TS^+(\gamma_+,\lambda_+,\alpha+)$, then its density $g^+(x,t)$ admits the Mellin-Barnes representation:
    \begin{equation}\label{density_TS+}
        g^+(x,t) =  -\frac{e^{ \lambda+^{\alpha_+} \mu_+ t \, - \, \lambda_+ x}}{\alpha_+ x } \int\limits_{c_2-i\infty}^{c_2+i\infty} \, \frac{\Gamma(1-s_2)}{\Gamma(1-\frac{s_2}{\alpha_+})} \, \left( \frac{-x}{(-\mu_+t)^{\frac{1}{\alpha_+}}}  \right)^{s_2} \, \frac{\ud s_2}{2i\pi}
    \end{equation}
    for any $c_2 < 1$.
\end{itemize}
\end{lemma}
\begin{proof}
It follows from \eqref{Laplace_exponent_TS} and from the Laplace inversion formula that:
\begin{equation}
    g^-(x,t) \, = \,  e^{ \lambda_-^{\alpha_-} \mu_- t }   \, \int\limits_{c_p-i\infty}^{c_p+i\infty} e^{-px} e^{-\mu_- t (\lambda_- + p)^{\alpha_-}} \, \frac{\ud p}{2i\pi}
\end{equation}
for $c_p>0$. From the frequency shifting property of the Laplace transform, we can write:
\begin{equation}
    g^-(x,t) \, = \,  e^{ \lambda_-^{\alpha_-} \mu_- t } \, e^{\lambda_- x } \, g_{\alpha_-}(x,t)
\end{equation}
where $g_{\alpha_-}(x,t)$ is the stable density \eqref{density_FMLS}, and (i) is proved. A similar approach can be used to prove (ii).\qed
\end{proof}

\subsection{Option pricing for negative tempered stable processes}
\label{subsec:pricing_tempered_densities}
Let $\alpha_-\in (1,2)$ and $X_t\sim TS^-(\gamma_-,\lambda_-,\alpha_-)$; from definition~\eqref{mu_def} and the Laplace exponent~\eqref{Laplace_exponent_TS}, the martingale adjustment reads:
\begin{equation}\label{mu_TS-}
    \mu \, = \, \left( (\lambda_- + 1)^{\alpha_-} - \lambda_-^{\alpha_-} \right) \mu_-
\end{equation}
where $\mu_-=-\gamma_-\Gamma(-\alpha_-)$ corresponds to the FMLS martingale adjustment \eqref{Laplace exponent}, and as expected $\mu\rightarrow\mu_-$ when $\lambda_-\rightarrow 0$. From the pricing formula \eqref{Risk-neutral_2} and using the notation \eqref{Mellin_G} and \eqref{Mellin_K}, the value at time $t$ of an option with maturity $T$ and payoff $\mathcal{P}(S_T,\underline{K})$ is equal to:
\begin{equation}\label{Risk-neutral:convolution_TS-}
    C_{\alpha_-,\lambda_-} (S,\underline{K},r,\mu_-,\tau)  =  e^{-(r - \lambda_-^{\alpha_-}\mu_- )\tau}  \int\limits_{c_1-i\infty}^{c_1+i\infty}  K_{\lambda_-}^*(s_1) G_{\alpha_-}^*(s_1) \, \left( -\mu_-\tau \right)^{-\frac{s_1}{\alpha_-}}  \frac{\ud s_1}{2i\pi}
\end{equation}
where we have defined
\begin{equation}\label{Mellin_K-}
    K_{\lambda_-}^*(s_1) \, := \, \int\limits_{-\infty}^{+\infty} \, e^{\lambda_- x} \, \mathcal{P} \left( S e^{(r+\mu)\tau +x} , \underline{K}  \right) x^{s_1-1} \, \ud x
    ,
\end{equation}
and where $\mu$ is given by \eqref{mu_TS-}. The $K_{\lambda_-}^*(s_1)$ function can be expressed in terms of the $K^*(s_1)$ function \eqref{Mellin_K} by introducing a Mellin-Barnes representation for the exponential term:
\begin{equation}
    K_{\lambda_-}^*(s_1) \, = \, \int\limits_{c_3-i\infty}^{c_3+i\infty} (-1)^{-s_3} \lambda_-^{-s_3} \Gamma(s_3) K^*(s_1-s_3) \frac{\ud s_3}{2i\pi}
\end{equation}
for $c_3>0$, and therefore, replacing in \eqref{Risk-neutral:convolution_TS-}, we obtain:

\begin{proposition}[Factorization]
\label{Prop:convolution_TS}
If $X_t\sim TS^-(\gamma_-,\lambda_-,\alpha_-)$ and if $\alpha_-\in (1,2)$, then the value at time $t$ of an option with maturity $T$ and payoff $\mathcal{P}(S_T,\underline{K})$ is equal to
\begin{multline}\label{Risk-neutral:convolution_TS-2}
    C_{\alpha_-,\lambda_-} (S,\underline{K},r,\mu_-,\tau) \, = \, e^{-(r - \lambda_-^{\alpha_-}\mu_- )\tau} \, \times \\
      \int\limits_{c_1-i\infty}^{c_1+i\infty} \int\limits_{c_3-i\infty}^{c_3+i\infty}
     (-1)^{-s_3} \lambda_-^{-s_3} \Gamma(s_3) K^*(s_1-s_3) G_{\alpha_-}^*(s_1) \, \left( -\mu_-\tau \right)^{-\frac{s_1}{\alpha_-}}  \frac{\ud s_1 \ud s_3}{(2i\pi)^2}
     .
\end{multline}
\end{proposition}

\paragraph{Example: digital power option (cash-or-nothing)} In that case, we know from \eqref{Mellin_K_DigitalCN} that:
\begin{equation}
    K^*(s_1-s_3) \, = \, - \frac{(-k_u- \rho_- \mu_- \tau)^{s_1-s_3}}{s_1-s_3}
\end{equation}
where $\rho_-:= \left( (\lambda_- + 1)^{\alpha_-} - \lambda_-^{\alpha_-} \right)$, and therefore it follows from \eqref{Risk-neutral:convolution_TS-2} that the digital cash-or-nothing call reads:
\begin{multline}\label{Call_C/N_TS-1}
    C_{\alpha_-,\lambda_-}^{(C/N)} (S,K,r,\mu_-,\tau) \, = \, \frac{1}{\alpha_-}e^{-(r - \lambda_-^{\alpha_-}\mu_- )\tau} \,    \int\limits_{c_1-i\infty}^{c_1+i\infty} \int\limits_{c_3-i\infty}^{c_3+i\infty} \, (-1)^{1-s_3} \, \times \\
    \frac{\Gamma(1-s_1)\Gamma(s_3)}{(s_1-s_3)\Gamma(1-\frac{s_1}{\alpha_-})}
    \lambda^{-s_3}
    (-k_u- \rho_- \mu_- \tau)^{s_1-s_3}
    \left( -\mu_-\tau \right)^{-\frac{s_1}{\alpha_-}}  \frac{\ud s_1 \ud s_3}{(2i\pi)^2}
    .
\end{multline}
The double integral \eqref{Call_C/N_TS-1} has a simple pole in $(0,0)$ with residue $1$, and a series of simple poles in $(1+n,m)$, $n,m\in\mathbb{N}$ induced by the singularities of the $\Gamma(1-s_1)$ and $\Gamma(s_3)$ functions. Summing all these residues yields:
\begin{multline}\label{Call_C/N_TS-2}
    C_{\alpha_-,\lambda_-}^{(C/N)} (S,K,r,\mu_-,\tau) \, = \, \frac{e^{-(r - \lambda_-^{\alpha_-}\mu_- )\tau}}{\alpha_-} \left[ 1 + \right. \\
    \left. \sum\limits_{\substack{n = 0 \\ m = 0}}^{\infty} \frac{(-\lambda_-)^m}{(1+n+m)n!m!\Gamma(1-\frac{1+n}{\alpha_-})} (k_u + \rho_- \mu_- \tau)^{1+n+m} \left( -\mu_-\tau \right)^{-\frac{1+n}{\alpha_-}} \right]
    .
\end{multline}
Note that when $\lambda_- = 0$, only the terms for $m=0$ survive and \eqref{Call_C/N_TS-2} degenerates into the $\alpha$-stable price \eqref{Call_C/N}, as expected. In the ATM forward case ($k_u$=0), the first few terms of the series \eqref{Call_C/N_TS-2} read:
\begin{equation}
 \frac{e^{-(r - \lambda_-^{\alpha_-}\mu_- )\tau}}{\alpha_-} \left[ 1 - \frac{\rho_-}{\Gamma(1-\frac{1}{\alpha_-})}(-\mu_-\tau)^{1-\frac{1}{\alpha_-}} + O \left( (-\mu_-\tau)^{2-\frac{2}{\alpha_-}} \right) \right]
\end{equation}
and can be Taylor-expanded for small $\lambda_-$:
\begin{equation}\label{Call_C/N_TS-approx}
    \frac{e^{-r\tau}}{\alpha_-} \left[  1 - \frac{(-\mu_-\tau)^{1-\frac{1}{\alpha_-}}}{\Gamma(1-\frac{1}{\alpha_-})} - \alpha_- \frac{(-\mu_-\tau)^{1-\frac{1}{\alpha_-}}}{\Gamma(1-\frac{1}{\alpha_-})} \lambda_- + O \left( \lambda_-^{\alpha_-}  \right) \right]
    .
\end{equation}
In the linear approximation \eqref{Call_C/N_TS-approx}, the intercept is the stable price, while the slope is governed by the negative left tail parameter $-\alpha_-$; the tempered stable price is therefore lower than the stable price (which is due to the tempering of the heavy tail), and the difference increases when $\alpha$ grows. This situation is displayed on fig.~\ref{fig:2} for $\alpha=1.7$, $K=4000$, $r=1\%$, $\sigma=20\%$, $\tau=2Y$, the series \eqref{Call_C/N} and \eqref{Call_C/N_TS-2} being truncated to $n_{max}=m_{max}=10$.

 \begin{figure}
  \centering
  \includegraphics[scale=0.25]{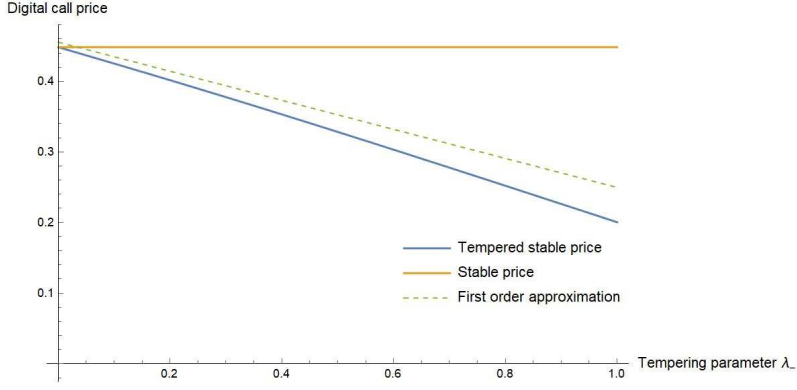}
 \caption{At the money stable \eqref{Call_C/N} and tempered stable \eqref{Call_C/N_TS-2} prices, and linear approximation \eqref{Call_C/N_TS-approx}; the tempered stable price intercepts the stable price when $\lambda_-=0$.}
 \label{fig:2}       
 \end{figure}

\section{Conclusions and future work}
\label{sec:conclusion}

In this article, we have derived generic representations in the Mellin space for path-independent options with arbitrary payoff, in the setup of exponential L\'evy models driven by spectrally negative stable or tempered stable processes. These representations have allowed us to obtain simple series expansions for the price of options with an exotic power-related payoff (Power Digital, Log or Gap Power options, Capped Power European options), by means of residue summation in $\mathbb{C}$ or $\mathbb{C}^2$. These series contain only simple terms and converge very fast, in particular when calls are in-the-money and for longer maturities; they can be very easily used for practical evaluation without requiring any help from numerical schemes.

Future work will include the investigation of path-dependent options, like Barrier or Lookback options; spectrally negative $\alpha$-stable processes are particularly interesting in this context, because the law of the supremum on a period of time is known to be \cite{Bingham73}:
\begin{equation}
    \mathbb{P} \left[ \underset{t\in[0,T]}{\mathrm{sup}} X_t \geq x \right] \, = \, \alpha \mathbb{P} \left[  X_T \geq x \right]
\end{equation}
which generalizes the reflection principle for the Wiener process ($\alpha=2$). Regarding path-independent instruments, the extension of the Mellin residue technique to two-sided L\'evy processes is currently in progress, with a particular focus on the Variance Gamma and the Normal Inverse Gamma (NIG) processes; the technique is very well adapted to these models too, because, like in the spectrally negative case, their density functions can be expressed under the form of Mellin integrals. Ongoing work also includes the extension of the pricing formulas obtained for the one sided tempered stable process in the case of digital (cash-or-nothing) options, to the more general case of European options.


%

%


\begin{acknowledgements}
The author thanks an anonymous reviewer for his/her insightful comments and suggestions.
\end{acknowledgements}

%
\section*{Conflict of interest}

The authors declares that he has no conflict of interest.



\end{document}